\documentclass[12pt]{amsart}
\usepackage{latexsym,amsmath,amsfonts,amscd,amssymb}
\usepackage{graphicx}
\usepackage{enumerate}
\newtheorem{theorem}{Theorem}[section]
\newtheorem{theorem*}{Theorem}

\newtheorem{proposition}[theorem]{Proposition}
\newtheorem{proposition*}[theorem*]{Proposition}
\newtheorem{corollary}[theorem]{Corollary}
\newtheorem{corollary*}[theorem*]{Corollary}

\newtheorem{definition}[theorem]{Definition}

\newtheorem{note*}[theorem*]{Note}

\theoremstyle{remark}

\newtheorem{remark*}[theorem*]{Remark}

\newcommand{\EE}{{\mathbb E}}

\newcommand{\NN}{{\mathbb N}}

\title{Three variations of Heads or Tails Game for Bitcoin}


\author[C. Grunspan]{Cyril Grunspan}
\address{Cyril Grunspan\newline{}\indent L\'eonard de Vinci P\^ole Univ, Finance Lab\newline{}\indent Paris, France, }
\email{cyril.grunspan@devinci.fr}
\author[R. P\'{e}rez-Marco]{Ricardo P\'{e}rez-Marco}
\address{Ricardo P\'{e}rez-Marco\newline{}\indent CNRS, IMJ-PRG, Univ. Paris Cit\'e \newline{}\indent Paris, France}
\email{ricardo.perez.marco@gmail.com}
\address{\tiny Author's Bitcoin Beer Address (ABBA)\footnote{\tiny Send some anonymous and moderate satoshis to support our research at the pub.}:
1KrqVxqQFyUY9WuWcR5EHGVvhCS841LPLn} 
\address{\includegraphics[scale=0.2]{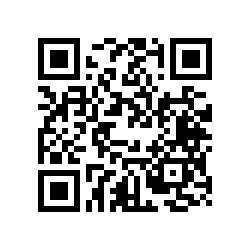}}

\begin{document}

\begin{abstract}
We present three very simple variants of the classic Heads or Tails game using chips, each of which contributes 
to our understanding of the Bitcoin protocol. The first variant addresses the issue of temporary Bitcoin forks, 
which occur when two miners discover blocks simultaneously. We determine the threshold at which an honest but 
temporarily ``Byzantine'' miner persists in mining on their fork to save his orphaned blocks. The second 
variant of Heads or Tails game is biased in favor of the player and helps to explain why the difficulty adjustment 
formula is vulnerable to attacks of Nakamoto's consensus. 
We derive directly and in a simple way, without relying on a Markov decision solver as was the case until now, 
the threshold beyond which a miner without connectivity finds it advantageous to adopt a deviant mining strategy on Bitcoin. 
The third variant of Heads or Tails game is unbiased and demonstrates that this issue in the Difficulty Adjustment formula 
can be fully rectified. Our results are in agreement with the existing literature that we clarify both qualitatively and quantitatively 
using very simple models and scripts that are easy to implement.
\end{abstract}

\date{May 05, 2024}

\maketitle

\section{Introduction}
The success of Bitcoin owes much to the simplicity of the Mathematics that explain its functioning 
(we refer to \cite{A14} for an introduction to the Bitcoin protocol,  and to \cite{GPM20-a} for some of the Mathematics involved). 
Drawing repeatedly on the concept of proof of work, these mathematics are essentially those arising from the game of Heads or Tails 
and particularly studied by the French mathematician and physicist of the 19th century, Siméon Poisson, whose name remains attached 
to the concept of Poisson processes. 
Thus, Poisson mathematics make it possible to calculate the exact probability of success of a 
double-spending attack \cite{GPM17, ROS14}. They also help understand why the current difficulty adjustment formula 
on Bitcoin allows so-called block withholding attacks \cite{GPM23}. In a way, the first to have thought of using a Heads or Tails 
game model to describe Bitcoin is Satoshi Nakamoto in his founding article \cite{N08}.

Considering game theory to describe certain aspects of Bitcoin is natural.
We refer the reader to a recent book by M. Warren on Bitcoin and game theory \cite{Warren2023}.

The goal of the present article is prove by direct and straightforward methods some results in the literature (including some that are not explicitly found). 
The threshold beyond 
which a miner no longer has an interest in staying honest on Bitcoin is at most equal to 32.94\% 
(and even less of course if the miner can rely on a non-zero connectivity parameter). This number is not 
found in the Bitcoin literature because most authors prefer to consider the problem where the miner's connectivity 
is by default equal to $\frac{1}{2}$ (see Subsection \ref{connectivity} for an explanation of connectivity for non-specialists). 
Nevertheless, one can in principle obtain it by running a Python implementation of the algorithm in the article \cite{SSZ17} and selecting $\gamma=0$ \cite{Mitsu}.
Also, it can be rigorously demonstrated that taking orphan blocks into account always makes the honest 
strategy the optimal one \cite{GPM23}.

We want to explain that the two problems were related. And indeed, this is easy to understand in a Heads or Tails game. 
If one does not always pay the dealer when the Flip is Heads, there is a bias in favor of the player. 
Moreover, one can intuitively understand why such a game can be made fair provided that the player is paid less in case of victory.

In this article we give a new perspective on certain old problems and we solve them with elementary tools:

$\bullet$ The first problem considered (Section \ref{temporarily}) is that of accidental forks. 
Such ``accidents'' still happen on Bitcoin despite a marked improvement in block propagation that takes 
only a few seconds to reach at least 50\% of the nodes \cite{DSN}. In this case, the network is divided until 
the discovery of a new block unambiguously defines the official blockchain. However, the creator of the orphaned block, 
who now has a block behind the official blockchain, may be tempted to start an attack and continue to mine on his orphaned block. 
We understand that there is a threshold in terms of relative hashing power beyond which this miner has an interest in 
persisting on his fork. This problem has been considered and studied in \cite{Kiayias16} under the name of ``The Immediate-Release Game''. 
The authors demonstrate two bounds for the threshold and then give a lower bound for the desired threshold. We show that 
this problem can be simply modeled using a Heads or Tails game with chips and can be solved with just a few lines of code. 

$\bullet$ The second problem addressed (Section \ref{permanent}) is the search for the highest hashrate for which 
the best mining strategy is the honest one when the connectivity is zero. In the general case, this problem 
is modeled as a Markov decision 
problem where the different states of the network are modeled by a triplet $(a,h,f)$ \cite{SSZ17}. However, when the 
connectivity is zero, the third parameter is irrelevant. States are described by just two parameters, exactly as 
in a modified Heads or Tails game with chips. The two parameters represent a number of chips for the player and the 
bank, which can eventually be converted into cash. The problem is easily solved if we require that the game end 
after a finite number of actions, which opens the way to a recurrence treatment for the search for the maximum 
expected gain under condition. In this game, the cost paid by the player-miner is proportional to the increase 
in the height of the official blockchain. So, the player-miner pays nothing if the result is favorable to him and 
he simply adds a new block to his secretly held fork. It is therefore easy to understand why the game is biased.

$\bullet$ However, if we take into account the production of orphan blocks in the difficulty adjustment formula 
(Section \ref{honesty}), which means for the player-miner having to pay each time the so-called ``difficulty'' 
function increases (which consists of the sum of official blocks and listed orphan blocks), the game is more balanced. 
Every time the player-miner wins, he must indeed pay for all the chips he crushes.

\section{Temporarily Byzantine by ``force of circumstance''}\label{temporarily}
When analyzing the security of certain systems, it is common practice in computer science to consider two very distinct 
categories of actors: honest participants, who respectfully follow the rules of protocol, and attackers. Following the 
terminology introduced by \cite{Lamport} in the study of distributed systems, the latter are called ``Byzantines''. 
In general, we don't change categories. Nevertheless, depending on the circumstances, we may occasionally be led to do so, such 
as a person who is fundamentally honest but finds a large sum of money on the street and decides to keep it for himself, without 
any effort to find the rightful owner.
We consider a simple situation where an honest miner on the Bitcoin network can be tricked into not respecting the rules 
of the protocol: the creation of a temporary fork. This is a relatively rare occurrence, but not an extremely rare one. 
According to statistical analyses carried out between 18/03/2014 and 14/06/2017, the rate of orphan block creation was 0.31{\%} 
for this period, and it is likely that this rate is even lower today thanks to the new versions of Bitcoin Core \cite{BSt, DSN}. 
We consider the case where two ``honest'' miners, each mining on the official blockchain, find a block at almost the same time. 
``Honest'' means here, as elsewhere in the article, that the miner always mines on the last block of the official blockchain and 
always immediately makes his discoveries public. In general, the first block discovered takes precedence and the second 
is considered ``orphaned'', although its terminology is imprecise. The Bitcoin wiki site prefers to speak of a 
``stale block'' \cite{BWi}. The miner who mined the second block is then drawn into a deviant posture. It is 
clearly in his interest to mine his orphan block rather than the last official block, because if he manages to mine a new 
block before the rest of the network, he will earn the reward contained in two blocks rather than just one. But then imagine 
that the other miners discover a block before he does. He must now not only catch up with the official blockchain, but also mine 
an additional block to gain the upper hand. Should he continue mining on his fork, or return to mining on the last block of the 
official blockchain, as stipulated by the Bitcoin protocol? This is an unprecedented situation for the miner, who eventually becomes 
``Byzantine'' by ``force of circumstances''.
The situation is indeed unprecedented, since it is perhaps the first time that it has been imposed on the miner, and it 
is unlikely that he will find himself in the same situation twice in a row in the future in the course of his activity. 
Furthermore, being fundamentally ``honest'', the miner, if he manages to catch up with and surpass the official 
blockchain, will benefit immediately. In particular, he will not engage in a block-witholding attack. The miner may become 
temporarily byzantine. His attack starts when he is one block behind the official blockchain and ends as soon as he gives 
up or manages to catch up and exceed the official blockchain by one block. In both cases, he resumes his position as an 
honest miner. The natural question is: Is it in his interest to continue mining on his fork, or should he abandon it and 
return to mining on the official blockchain? What is the threshold in terms of relative hash power at which an a priori ``honest'' 
miner has an interest in stubbornly mining on his fork when he is one block behind on the official blockchain? This question 
can be resolved using a simple classic coin toss, which we'll now describe. 

\subsection{A first classic variation of Heads or Tails}
This game pits a player against a bank. Over time, the player and the bank gain or lose chips (corresponding to mined blocks) and 
can exchange them for euros (corresponding to bitcoins). At any given moment, the player has three possible actions, as described below.

\begin{description}
\item[Toss] A croupier tosses a coin rigged in favor of the bank. The probability of getting Tails is $q$. 
This action costs the player $q$ euros whatever the result.
\begin{itemize}
\item If the result is Heads, the bank wins a chip.
\item  If the result is Tails, the player wins a chip.
\end{itemize}
\item[Crush] This action is only possible if the player (resp. the bank) has $a$ (resp. $h$) chips 
with $a>h$. In this case, the bank loses all its chips, the player loses $h+1$ chips but gains $h+1$ euros. This action costs the player nothing. 
\item[Abandon] The bank and the player lose all their chips. This action costs the player nothing.
\end{description}

In addition, the player has only a finite number $n$ of possible actions, and the game ends immediately as 
soon as the player uses either the Crush action or the Abandon action. The number of actions available to the player 
is therefore essentially the number of Toss actions, and the game ends as soon as the player gives up or takes advantage of the bank, corresponding to $a=h+1$. 
This is not an attack on bitcoin's difficulty adjustment formula, as a block-witholding attack might be. 
The player-miner is profoundly "honest". We therefore consider that we always have a priori $n\leq 2016$ (since a difficulty 
adjustment takes place every 2016 blocks on Bitcoin). Let's describe each action.

\begin{itemize}
\item The Toss action corresponds to the fact that the miner persists in mining on his fork despite a delay on the official blockchain. 
\item The Crush action corresponds to the fact that the miner has a secret fork enabling him to gain an advantage over the official 
blockchain. He then decides to make it public and pocket all the rewards it contains; he then naturally resumes his position as an honest miner.
\item The Abandon action means that the miner returns to mining on the last block of the official blockchain, like any honest miner.
\end{itemize}

In the temporary fork situation we're considering, a mining strategy is just the stopping time $\xi$ (bounded by 2016) that 
specifies the first instant when the miner returns to mine on the official blockchain. We denote $R(\xi)$ the income earned 
by the miner following this strategy. We need to compare $R(\xi)$ with the income the miner would have earned in $\xi$ if he had mined 
honestly all along. Given that the miner's relative hashing power is $q$, the latter quantity is worth on
average $q b\frac{\EE[\xi]}{\tau_0}$ with ${\tau_0}=10$ minutes and $b=3.125$ BTC (current value of a coinbase) plus the average 
value of transaction fees present in a block. So the key quantity for choosing the $\xi$ strategy over the honest one 
is $\EE[R(\xi)]-q b \frac{\EE[\xi]}{\tau_0}$. The $\xi$ strategy is preferable to the honest strategy if $\EE[R(\xi)]-q b \frac{\EE[\xi]}{\tau_0}>0$.
The second term $-q b \frac{\EE[\xi]}{\tau_0}$ is then interpreted as a cost. In this expression, everything happens 
as if the miner were paying $q b$ every time a block is discovered. Hence the fact that the player-miner pays a fixed cost to the croupier, 
which is $q$ whatever the result of the coin toss. The parameter $q$ is the probability that the coin will land on Tails, which corresponds 
to the miner finding a block before the honest miners.

\begin{definition}
Let $JM_1(a,h)$ be the Heads or Tails game described above, where the player starts from an initial situation 
in which he has a chip against the bank, which has h chips. Note also $E_1(a,h,n,q)$, the maximum expected pay-off 
for a player with a maximum of $n$ possible actions.
\end{definition}

\begin{proposition}
  The game $JM_1 (0, 0)$ is a fair game.
\end{proposition}

\begin{proof}
  Suppose the player has chosen a strategy and let $J_1$ (resp.
  $J_2$) be the number of average chips received in total by the player (resp.  the bank) by participating in the game. 
  The coin toss  by the croupier is biased in favor of the bank, we have on average: 
  $J_1 = qJ$ and $J =  J_1 + J_2$. 
  Now $qJ$ is the total cost paid by the player in euros and  
  $J_1$ is the maximum sum in euros received by the player 
  (in the ultra-favorable  case where he succeeds in converting all his chips into  euros). 
  Therefore, the player's maximum net pay-off is zero. 
\end{proof}

However, the $JM_1(a,h)$ game with $a>0$ may be biased. By assumption, we have: 
\begin{eqnarray}
  E_1 (a, h, 0, q) & = & 0 
\end{eqnarray}
and for all $n \in \mathbb{N}^{\ast}$ and $a > h$,
\begin{eqnarray}
  E_1 (a, h, n, q) & = & a 
\end{eqnarray}
On the other hand, if $a \leqslant h$, the player-miner has the choice of continuing
mining (action \textbf{Toss}) or giving up (action
\textbf{Abandon}). So, for any $n \in \mathbb{N}^{\ast}$ and $a \leqslant
h$,

\begin{align}
  E_1 (a, h, n) & =  \max \{ 0, q \cdot E_1 (a + 1, h, n - 1, q) + (1 -
  q) \cdot E_1 (a, h + 1, n - 1, q) - q \} 
\end{align}
Below is a simple pseudo-code that accurately gives the average maximum gain $E_1 (a, h, n, q)$. 
We use the memoization principle for the sake of efficiency.

 \bigskip
{\tt

function E1(a, h, n, q, memo):
\par \hskip 10pt     if (a, h, n) in memo:
\par \hskip 20pt         return memo[(a, h, n)]

\par \hskip 10pt     if n == 0:
\par \hskip 20pt         return 0

\par \hskip 10pt     if a > h:
\par \hskip 20pt         return a

\par \hskip 10pt     memo[(a, h, n)] = max(0, q * E1(a + 1, h, n - 1, q, memo) + (1 - q) 
\par \hskip 20pt      * E1(a, h + 1, n - 1, q, memo) - q)
     
\par \hskip 10pt     return memo[(a, h, n)]
}
 \bigskip
 
 (maximum average net income $E_1$)
 \bigskip

We observe that $E_1 (1, 2, 75, 0.429056) = 4.050134694288943 \times 10^{-8} > 0$. However, we are unable 
to find $n$ such that $E_1 (1, 2, n, 0.429055) > 0$. In other words, if $q > 42, 91$\%, the game $JM_1 (1, 2, n)$ is
in favor of the player for $n$ large enough. Therefore, in the case of a
temporary fork, the minimum threshold beyond which a miner a priori honest miner
to insist on mining on his fork even though he's one block behind on the blockchain is about $42, 91$\%. 
The result is in line with the $36.1$\% and $45.5$\% bounds obtained in \cite{Kiayias16} with calculus in their section 
``The immediate-Release Game''. Moreover, in a presentation of this article given at the University of 
Crete in $2019$ the speaker (who is also one of the authors of the article) is more precise 
and states that the threshold lies between $42$\% and $43$\% (the lower bound $42$\% is also in the paper) \cite{Kiayias16, Koutsoupias19}.

\section{To be or not to be totally Byzantine?}\label{permanent}
We now consider another mining problem. At what relative hash power $q$ does it no longer make sense for a miner to be honest? 
For a miner, being honest means always mining on the last block of the official blockchain and always making any blocks 
discovered public. Not being honest means the opposite: keeping discovered blocks secret or not mining on the last block of the 
official blockchain. The problem under consideration is fundamentally different from the one previously considered. The miner is not 
an honest miner who momentarily becomes ``Byzantine'' by force of circumstances. On the contrary, he chooses his camp from the 
outset (honest or Byzantine) and never leaves it. What's more, his mining strategy is not limited in time. On the contrary, it is unlimited and repetitive. 
This problem has certainly already been solved in the general case \cite{SSZ17}. The authors recognize a Markov decision 
problem which they solve using a solver. They are then confronted with technical issues, as a priori such a solver can only be used in the 
case where the number of system states is finite. We show that it is possible to simplify this problem in the case where the connectivity 
of the miner is zero, and that it then resembles a simple Heads or Tails problem that we solve simply without using a solver. Recall that 
miner connectivity is a parameter introduced by \cite{ES14} and picked up by many authors since then. For non-specialists, let's describe 
the connectivity parameter here.

\subsection{Connectivity}\label{connectivity}
Let's imagine that a miner (an attacker) with a relative hash power $q$ has already mined a block $A$ of height $H$, but is keeping 
it secret because he is following a deviant mining strategy of block withholding. If the miner learns from one of his neighboring nodes 
that a block $B$ of the same height $H$ has just been discovered, he can retaliate by immediately making his hidden block A public. In this case, 
part of the network will learn of $A$'s existence before $B$'s, and vice versa. The network will be divided, with some honest miners 
looking for a block above $A$ and others above $B$.

Connectivity, denoted $\gamma$, is a measure of the amount of computing power the attacker diverts to mine his block A. 
Mathematically, it is defined as the probability that block $A$ will find its way onto the official blockchain, knowing 
that honest miners will find a block of height $H+1$. This is a conditional probability.

In concrete terms, in the situation described above, following the discovery of a block of height $H+1$, following on from blocks A and B, 
we find ourselves in one of the following three cases:

\begin{itemize}
\item the attacker is the first to discover a block of height $H+1$ ;
\item the rest of the network discovers a block of height $H+1$ above $A$ before the attacker;
\item the rest of the network discovers a block of height $H+1$ above $B$, before the attacker.
\end{itemize}

Note that in the second case, block $A$ is somehow saved by the honest miners who received $A$ before $B$. 
The probability of occurrence of the first event is $q$, that of the second is 
$\gamma p$ and that of the third is $(1-\gamma) p$ with $p=1-q$.

The parameter $\gamma$ measures the attacker's ability to react. If he is well connected, he will quickly 
learn of the existence of a new block before the others and announce the existence of his own hidden block to 
the rest of the network. This is a measure of his ability to create confusion in the network.

Since the network is constantly evolving, it's an illusion to believe that $\gamma$ remains constant over time. 
However, this is an assumption often made when assessing the profitability of mining strategies.

\subsection{A simplified problem when connectivity is zero}
In itself, connectivity is an attack vector that was not imagined by Satoshi Nakamoto, since it does not feature in 
his founding paper. With $\gamma=1$, a miner has no incentive to be honest. He has no interest in publishing a block he has 
just discovered. He can simply wait for another block to be discovered and react then. It is interesting to pose 
$\gamma=0$ to understand how Nakamoto's consensus can naturally be faulted without adding this attack vector. We therefore 
formulate this hypothesis ($\gamma=0$) and, within this framework, we seek to find out whether a miner has an incentive 
to behave honestly or whether there is a more profitable mining strategy. A mining strategy specifies the action to be taken by 
the miner depending on the state of the network. The chosen mining strategy, whether honest or not, is repetitive. In due course, 
the miner will almost certainly return to his starting point and mine on the last block of the official blockchain. When this happens, 
the miner is said to have completed a cycle. During this cycle, we note $R$ the number of blocks added by the attacker to the official blockchain 
and $H$ the progression of the height of the official blockchain. In concrete terms, at any given moment, the attacker 
has the choice between mining on his secret fork, overriding the official blockchain when he has the means to do so, or giving 
up and returning to mine on the last block of the official blockchain. In reality, in the general 
case, he has an additional action at his disposal: the action noted as ``Match'' by the authors \cite{SSZ17}, which consists 
in revealing a block already mined but kept secret by the attacker. However, under our assumption $\gamma=0$, this action is not possible. 
Furthermore, the state of the network is normally modeled by a triplet $(a,h,f)$ where $a$ (resp. $h$) designates the number 
of blocks mined by the attacker (resp. honest miners) on the last fork created and $f$ designates the possibility of using the Match action 
or the fact that it is already activated. In the case of $\gamma=0$, the latter parameter is irrelevant. 
The state of the network is simply modeled by a pair $(a,h)$. 

\subsection{The effect of the difficulty adjustment formula}
PnL (Profit and Loss) per unit of time is the only quantity that makes economic sense. The cost of mining per unit of time 
is independent of the mining strategy chosen (keeping blocks secret may have an impact on the miner's income, but not on his 
cost of mining). So, the quantity that allows us to compare different full time mining strategies is the revenue per unit of time. 
This key observation  was made in \cite{GPM2018a}, before only the relative proportion of mined blocks in the official blockchain 
was used, without justification, as a benchmark of profitability of the strategy. Only in the long run, after difficulty adjustments these are equivalent.
A difficulty adjustment formula occurs when the official blockchain grows by $2016$ blocks. This has the effect of 
maintaining an average duration of $10$ minutes each time the height function of the official blockchain increases by $1$. 
In these circumstances, only in the long run, the percentage of blocks mined by the attacker present in the official 
blockchain gives the Revenue in the long run. In concrete terms, when a given mining strategy, or ``minning policy'', is modeled by a Markov chain as in \cite{SSZ17}, 
only when the interblock time stabilizes, the measure of revenue per unit of time  is given by $\frac{\EE[R]}{\EE[H]}$ for integrable strategies, i.e.
those with finite time expectations for the cycles for which $\EE[H]<\infty$.
When the miner mines honestly, this quantity is equal to the miner's relative hash 
power, which we have always denoted $q$. 
This leads to the following proposition (Corollary 10.1 of \cite{SSZ17}, which is a Corollary of the more general Proposition 3.6 of \cite{GPM2018a}). 

\begin{proposition}\label{profit2}
An admissible mining strategy is more profitable than the honest strategy if and only if $\EE[R-q H]>0$.\label{prop1}
\end{proposition}

As in the previous section, we can interpret the second term (here -$q H$) as a cost. But unlike in the previous section, the miner 
no longer pays $q$ each time a block is discovered (by him or the rest of the network), but only each time the official blockchain progresses. 
This leads us to consider another version of the Heads or Tails game.

\subsection{Another Heads or Tails game }

During the course of the game, the player regularly accumulates chips and can, under certain constraints, convert them into 
cash (euros, let's say). At any given moment, the player has a maximum of three possible actions.

\begin{description}
\item[Toss] A croupier tosses a coin rigged in favor of the bank. The probability of getting Tails is $q$.
\begin{itemize}
\item If the result is Tails, the player wins a chip and pays nothing.
\item  If the result is Heads, the bank wins a chip and the player pays $q$ to the dealer.
\end{itemize}
\item[Crush] This action is only possible if the player (resp. the bank) has $a$ (resp. $h$) chips with $a>h$. In this case, 
the bank loses all its chips, the player loses $h+1$ chips but wins $h+1$ euros and also gives $q$ euros to the dealer. 
His net result is therefore $h+1-q$ euros.
\item[Abandon] The bank and the player lose all their chips. This action costs the player nothing.
\end{description}

There are several differences with the game studied in the previous section. Firstly, the Crush and Abandon actions 
do not end the game. The player can use them and continue playing if he has enough actions available. Secondly, the 
Crush action does not earn exactly $h+1$ euros as before, but $h+1-q$ euros. Last but not least, the player doesn't 
always pay the dealer! This is only the case when the toss is unfavorable, i.e. when the result is Heads and the banker 
wins a chip. If the result is Tails, he wins a chip as before, but gives nothing!

Here again, a few comments are in order.

\begin{itemize}
\item A coin toss by the dealer is equivalent to the discovery of a block by the honest miner or miners.
\item The Toss action is equivalent to the miner choosing to mine secretly and wait for a block to be discovered. 
If the result is Heads, then the official blockchain advances by one block. The height function therefore increases by 1, 
resulting in a cost q paid by the player in this case. If, on the other hand, the result is Tails, the block discovered by the miner is kept secret. 
The height of the official blockchain does not increase. Hence a zero cost.
\item The Crush action means that the miner replaces the last $h$ blocks of the official blockchain with his own. 
For this action to be possible, the miner must reveal one more block ($h+1$ in all), which increases the height of the 
official blockchain by $1$. The miner then gains the reward contained in $h+1$ blocks, and at the same time, the height 
of the official blockchain increases by $1$. Hence a net gain of $h+1-q$.
\item The Abandon action is equivalent to the miner dropping his secret fork and returning to mine on the last block of 
the official blockchain. He neither gains nor loses anything with this action, as the height of the official blockchain remains unchanged.
\end{itemize}

\begin{definition}
Let $JM_2(a,h)$ denote the previous game in which the player starts from a situation in which he owns a chips against the bank, 
which owns $h$ chips. Let $E_2(a,h,n,q)$ also be the maximum payoff expectation of the player starting from an initial situation 
in which he has $a$ chips and the bank has $h$, under the assumption that he has at most $n$ possible actions.
\end{definition}
Proposition \ref{profit2} can be reinterpreted.

\begin{proposition}
There is a more profitable strategy than the honest one if and only if the $JM_2(0,0)$ game is biased in the player's favor.
\end{proposition}

We have the following relations $(a, h) \in \mathbb{N}^2$ and $q \in [0, 0.5[$,
\begin{eqnarray}
  E_2 (a, h, 0, q) & = & 0 
\end{eqnarray}
For any $n \in \mathbb{N}^{\ast}$ and $a > h$, if the player decides to use the
to use the Crush action, then the state of the network changes from $(a, h)$ to $(a - h - 1, 0)$ and the game continues. Thus,

\[
\begin{aligned}
E_2 (a, h, n, q) = \max \Bigg\{ &(h + 1) - q + E_2 (a - h - 1, 0, n - 1, q), \\
& q \cdot E_2 (a + 1, h, n - 1, q) + (1 - q) \cdot (E_2 (a, h + 1, n - 1, q) - q) \Bigg\}
\end{aligned}
\]

In the case where $n \in \mathbb{N}^{\ast}$ and $a \leqslant h$, the player has the choice between continuing 
to mine (action Launch) or abandoning (action Abandon). This latter action does not end the game but leads the player to the state $(0,0)$ 
with one less action. Therefore, for all $n \in \mathbb{N}^{\ast}$ and $a \leqslant h$,

\[
\begin{aligned}
E_2 (a, h, n, q) & = \max \Bigg\{ E_2 (0, 0, n - 1, q), \\
& q \cdot E_2 (a + 1, h, n - 1, q) + (1 - q) \cdot (E_2 (a, h + 1, n - 1, q) - q) \Bigg\}
\end{aligned}
\]

The advantage of this approach is that we can avoid relying on a Markov decision solver. Below is a very 
simple pseudo-code that precisely provides the maximal expected gain $E_2(a,h,n,q)$ through memoization.

 \bigskip
  {\tt 
  
  function E2(a, h, n, q, memo): 
  
\par \hskip 10pt     if (a, h, n) in memo:
 
\par \hskip 20pt           return memo[(a, h, n)]

\par \hskip 10pt     if n == 0:

\par \hskip 20pt          memo[(a, h, n)] = 0

\par \hskip 10pt          return 0
     
\par \hskip 10pt      if a > h:
\par \hskip 20pt         memo[(a, h, n)] = max((h + 1) - q + E2(a - h - 1, 0, n - 1, q, memo),
\par \hskip 20pt                                q * E2(a + 1, h, n - 1, q, memo) + (1 - q) * (E2(a, h + 1, n - 1, q, 
\par \hskip 20pt         memo) - q))
\par \hskip 10pt      else:
\par \hskip 20pt         memo[(a, h, n)] = max(E2(0, 0, n - 1, q, memo),
\par \hskip 20pt                                q * E2(a + 1, h, n - 1, q, memo) + (1 - q) * (E2(a, h + 1, n - 1, q, 
\par \hskip 20pt        memo) - q))
\par
\par \hskip 10pt      return memo[(a, h, n)]
}
\bigskip

(Maximum average net income $E_2$)

\bigskip

Note that $E_2(0,0,146,0.329393)=4.4530581139179404\times 10^{-8}>0$. 
However, it is impossible to find an integer $n$ such that $E_2(0,0,n,0.329392)>0$. Therefore, the threshold beyond 
which a miner with zero connectivity is incentivized to choose a deviant strategy is approximately $32.94\%$. This result coincides 
with that obtained with the Python implementation of the article \cite{SSZ17}, see \cite{Mitsu}. Note that this threshold is not far 
from the $\frac{1}{3}$ threshold for the classical selfish mining strategy. 
This also suggests that the selfish mining strategy is not far from the optimal strategy for low $q$ even 
if  it doesn't highlight the fact that sometimes the miner may have to continue his attack while lagging behind the official blockchain, see  \cite{KMNS16} and 
\cite{GPM2018b} for the analysis of these Stubborn mining strategies. 

\section{Honesty is the best policy.}\label{honesty}
In this section, we consider the case where the Bitcoin difficulty adjustment formula has been modified, and we 
theoretically demonstrate, based on a variation of the Heads or Tails game, that the best strategy is still the honest strategy when 
$\gamma = 0$. The general result without assumptions on $\gamma$ has also been proven in \cite{GPM23}

\subsection{A new difficulty adjustment formula}
Today, the nodes in the Bitcoin network do not transmit orphaned blocks. But let's imagine that they could. Let's even imagine that 
miners are incentivized to do so by modifying the rule that governs the official blockchain. In the event of a 
tie between two blockchains, we should select the one that contains the most proof of work, taking into account orphaned blocks as 
well, provided that they have an ancestor in the considered blockchain. This would be a kind of reinforcement of Satoshi Nakamoto's rule. 
In case of another tie between two blockchains with the same characteristics, a node would select, as it does today, the one that was transmitted to it first.

In this case, the key quantity for comparing two mining strategies would no longer be 
$\frac{\EE[R]}{\EE[H]}$ but $\frac{\EE[R]}{\EE[D]}$ where $R$ would still represent the number of blocks 
added to the official blockchain by the miner during a cycle and $D$ would represent the progression of the 
so-called difficulty function during the same period (authors in \cite{BET20} introduced the 
concept of ``difficulty function''). We would have $D=H+U$ where $U$ is the number of orphan blocks mined during a 
cycle. The effect of the new difficulty adjustment algorithm would be to impose a duration of $10$ minutes on each progression 
of the difficulty function (instead of the height function as now). As before, this would mean that a mining 
strategy would be more profitable than the honest strategy if and only if $\EE[R-q D]>0$ where $D$ here represents the progression 
of the difficulty function over a cycle, which leads us to consider another game of Heads or Tails.

\subsection{A third Coin Toss game}
During this game, the player regularly accumulates chips that, under certain constraints, they can 
convert into cold hard cash (let's say euros). At any given moment, the player has at most three possible actions:

\begin{description}
\item[Toss] A croupier tosses a coin rigged in favor of the bank. The probability of getting Tails is $q$.
\begin{itemize}
\item If the result is Tails, the player wins a chip and pays nothing.
\item  If the result is Heads, the bank wins a chip and the player pays $q$ to the dealer.
\end{itemize}
\item[Crush] This action is only possible if the player (or the bank) has $a$ (or $h$) chips with $a>h$. In this case, 
the bank loses $h$ chips, the player loses $h+1$ chips and wins $h+1$ euros, but he must also give the dealer $q\,(h+1)$.
Hence, their net result is $(1-q).(h+1)$ euros.
\item[Abandon] The bank and the player lose all their chips. This action costs the player nothing.
\end{description}

\begin{definition}
We denote by $JM_3(a,h)$ the game described above, where the player starts from a situation where they have $a$ 
chips against $h$ for the bank, $(a,h)\in\NN^2$, and let $E_3(a,h,n)$ be the maximal net income of the player playing $JM_3(a,h)$ 
and having at most $n$ actions available.
\end{definition}

The only difference between $JM_2$ and $JM_3$ is the consequence of the Crush action. In $JM_3$, the player earns 
less than in $JM_2$. Indeed, in $JM_3$, the player must pay for the official blockchain to advance by 1 as in $JM_2$, 
but he must also pay for the creation of $h$ orphan blocks (the $h$ blocks of honest miners that have been replaced and 
made visible to all). Hence a cost equal to $q.(h+1)$ as a result of this action and therefore the fact that the gain 
is now only $(1-q).(h+1)$ euros which is clearly less than $h+1-q$ euros as in the previous game.

\begin{proposition}
The $JM_3(0,0)$ game is a fair game.
\end{proposition}

\begin{proof}
Suppose the player has chosen a strategy and let $J_1$ (resp. $J_2$) be the average number of chips received in total 
by the player (resp. the bank) while participating in the game. As the coin tossed regularly by the dealer is biased in favor 
of the bank, we have on average: $J_1=q J$ and $J_2=p J$ with $p=1-q$ and $J=J_1+J_2$. Therefore, $p J_1=q J_2$. 
Now, $p J_1$ is the maximum sum received by the player in euros (i.e. the most he wins in the ultra-favorable case where he 
manages to convert all his chips into euros) and $q J_2$ is the cost in euros paid by the player. Therefore, the 
player's maximum net gain expectation is zero. 
\end{proof}
We can be more precise and show the following result.
\begin{theorem}
For all integers $a, h,n$, we have $E_3 (a, h, n) \leqslant p \cdot a$ with $p  = 1 - q$.
\end{theorem}

\begin{proof}
The result is true if $n=0$. Assume it to be true at rank $n-1\geqslant 0$. Then,
\begin{equation*}
qE_3 (a + 1, h, n - 1) + p \cdot (E_3 (a, h + 1, n - 1) - q) \leqslant q \cdot p \cdot (a + 1) + p \cdot (p \cdot a - q) = p \cdot a
\end{equation*}
  So, if $a > h$,

\begin{equation*}
\begin{aligned}
E_3 (a, h, n) & = \max \Bigg\{ (h + 1) \cdot (1 - q) + E_3 (a - h - 1, 0, n - 1), \\
& \quad qE_3 (a + 1, h, n - 1) + (1 - q) \cdot (E_3 (a, h + 1, n - 1) - q) \Bigg\} \\
& \leqslant \max \{ p (h + 1) + p (a - h - 1), pa \} = pa
\end{aligned}
\end{equation*}

  and if $a \leqslant h$,
\begin{equation*}
\begin{aligned}
E_3 (a, h, n) & = \max \Bigg\{ E_3 (0, 0, n - 1), \\
& \quad qE_3 (a + 1, h, n - 1) + (1 - q) \cdot (E_3 (a, h + 1, n - 1) - q) \Bigg\} \\
& \leqslant \max \{ 0, pa \} = pa
\end{aligned}
\end{equation*}

  Hence we get the result.
\end{proof}
Hence the corollary which confirms the previous result
\begin{corollary}
  For all $n\in\NN$, $E_3 (0, 0, n, q) = 0$.
\end{corollary}
Thus, the game of $JM_3(0,0)$ modified mining is unbiased, corresponding to the fact that the best mining 
strategy is the honest one. The result has been demonstrated for $\gamma=0$ but can also be demonstrated for any $\gamma$ using 
more powerful tools with martingales \cite{GPM23}.

\section{Conclusion}
In this article, we aim to calculate different thresholds in terms of relative hash power, 
beyond which a miner might be tempted to engage in a deviant mining strategy. In each case, we use a simple 
coin-toss model. Each time, we calculate the player's maximum expected payoff under the constraint of a limited 
number of possible actions, using a very simple script. The results are classic. This approach also allows us to 
qualitatively understand why the current difficulty adjustment formula in Bitcoin is flawed and opens the door to potential 
attacks. It also demonstrates how this problem can be corrected.


\begin{thebibliography}{1}

  \bibitem[1]{A14}A.~Antonopoulos. \textit{Mastering Bitcoin: Unlocking Digital Cryptocurrencies} O’Reilly Media, Inc., 2014.

  \bibitem[2]{BET20}R.~Bar-Zur, I.~Eyal, A.~Tamar. \textit{Efficient MDP analysis for Selfish-Mining in Blockchains.} Proceedings of the 2nd ACM Conference on Advances in Financial Technologies, 2020.


  \bibitem[3]{BWi}Bitcoin Wiki. \textit{Orphan Block.} https://en.bitcoin.it/wiki/Orphan{\_}Block.
  
  \bibitem[4]{BSt}Bitcoin Stackexchange. \textit{What are orphaned and stale blocks?} https://bitcoin.stackexchange.com/questions/5859/what-are-orphaned-and-stale-blocks  

 \bibitem[5]{DSN}Decentralized Systems and Network Services Research Group - KASTEL. \textit{Bitcoin Monitoring} https://www.dsn.kastel.kit.edu/bitcoin
  
  \bibitem[6]{ES14}I.~Eyal, E.~Sirer. \textit{Majority is not
  enough: bitcoin mining is vulnerable.} International
  Conference on Financial Cryptography and Data Security, 2014.
  

  \bibitem[7]{GPM17}C.~Grunspan  and  R.~P\'erez-Marco. \textit{Double spend races}, International Journal of Theoretical and Applied Finance, Vol. 21, 2018.
  
  \bibitem[8]{GPM2018a}C.~Grunspan  and  R.~P\'erez-Marco. \textit{On profitability of Selfish Mining}, 
  ArXiv:1805.08281, 2018.
  
  \bibitem[9]{GPM2018b}C.~Grunspan  and  R.~P\'erez-Marco. \textit{On profitability of Trailing Mining}, 
  ArXiv:1811.09322, 2018.
  
  
  \bibitem[10]{GPM20-a}C.~Grunspan, R.~P{\'e}rez-Marco. \textit{The Mathematics of Bitcoin.} Newsletter of the 
  European Mathematical Society Newsletter, \textbf{115}, p.31-37, 2020.

  \bibitem[11]{GPM23}C.~Grunspan, R.~P{\'e}rez-Marco. \textit{Blockchain Witholding Resilience.} https://arxiv.org/abs/2211.07270, 2023.
  
  \bibitem[12]{Kiayias16}A.~Kiayias, E.~Koutsoupias, M.~Kyropoulou, Y.~Tselekounis. \textit{Blockchain mining games.} Proceedings of the 2016 
  ACM Conference on Economics and Computation (pages 365--382), 2016
  
  
  \bibitem[13]{Koutsoupias19}E.~Koutsoupias. \textit{Talk by Elias Koutsoupias at ECE TUC (see between 28' and 30').} https://www.youtube.com/watch?v=K6iTBLhsFA0, 2019
    
  
  \bibitem[14]{Lamport}L.~Lamport,R.~Shostak, M.~Pease \textit{The Byzantine Generals Problem.} ACM Transactions on Programming Languages and Systems, 1982
     
     
  \bibitem[15]{Mitsu}Mitsuhamizu. \textit{A python implementation for solving the MDP in Optimal selfish mining.}   
https://github.com/Mitsuhamizu/Optimal-selfish-mining     
 
  \bibitem[16]{N08}S.~Nakamoto. \textit{Bitcoin: a peer-to-peer
  electronic cash system.} Bitcoin.org/bitcoin.pdf,
  2008.
  
  \bibitem[17]{KMNS16}K.~Nayak, E.~Shi, S.~Kumar, A.~Miller.
  \textit{Stubborn mining: generalizing selfish mining and combining with an eclipse attack.} IEEE European Symp.  Security
  and Privacy,  pages  305--320, 2016.
  
  \bibitem[18]{ROS14}M.~Rosenfeld. \textit{Analysis of hashrate-based double spending.} arXiv:1402.2009, 2009.    


  \bibitem[19]{SSZ17}A.~Sapirshtein, Y.~Sompolinsky, A.~Zohar. \textit{Optimal Selfish Mining Strategies in Bitcoin.} International Conference on Financial 
  Cryptography and Data Security, 2017.  
  
  \bibitem[20]{Warren2023}M.~Warren. \textit{Bitcoin: A Game-Theoretic Analysis}, De Gruyter Graduate, 2023.    

  
  
  
\end{thebibliography}
\end{document}